\documentclass[sigconf]{acmart}

\pdfoutput=1

\usepackage{verbatim} % provides a comment environment
\usepackage{paralist}

\usepackage{ifthen}

\usepackage[ruled,vlined,linesnumbered]{algorithm2e}

\newcommand{\Dfn}[1]{\textbf{\emph{#1}}}
\newcommand{\Dom}[1]{\ensuremath{\mathit{domain}(#1)}}
\newcommand{\Ran}[1]{\ensuremath{\mathit{range}(#1)}}

\newcommand{\TARGETDOC}{0}

\begin{document}
\title{Quantifying Policy Administration Cost in an Active
  Learning Framework}

%\ifthenelse{\TARGETDOC = 0}{}{}

\ifthenelse{\TARGETDOC = 0}{
\author{Si Zhang}
\affiliation{
    \institution{Department of Computer Science \\ University of Calgary}
    \city{Calgary}
    \country{Canada}
}
\email{si.zhang2@ucalgary.ca}

\author{Philip W. L. Fong}
\affiliation{
    \institution{Department of Computer Science \\ University of Calgary}
    \city{Calgary}
    \country{Canada}
}
\email{pwlfong@ucalgary.ca}
}{
  \author{Anonymized}
}

\begin{abstract}
  This paper proposes a computational model for policy administration.
  As an organization evolves, new users and resources are gradually
  placed under the mediation of the access control model. Each time
  such new entities are added, the policy administrator must
  deliberate on how the access control policy shall be revised to
  reflect the new reality. A well-designed access control model must
  anticipate such changes so that the administration cost does not
  become prohibitive when the organization scales up.  Unfortunately,
  past Access Control research does not offer a formal way to quantify
  the cost of policy administration.  In this work, we propose to
  model ongoing policy administration in an active learning
  framework. Administration cost can be quantified in terms of query
  complexity.  We demonstrate the utility of this approach by applying
  it to the evolution of protection domains.  We also modelled
  different policy administration strategies in our framework.  This
  allowed us to formally demonstrate that domain-based policies have a
  cost advantage over access control matrices because of the use of
  heuristic reasoning when the policy evolves.  To the best of our
  knowledge, this is the first work to employ an active learning
  framework to study the cost of policy deliberation and demonstrate
  the cost advantage of heuristic policy administration.
\end{abstract}




\maketitle

\section{Introduction}
Access Control is concerned with the specification and enforcement of
policies that govern who can access what.  Access control policies,
however, must be revised when the organization's needs evolve. A
typical situation that motivates changes to an existing access control
policy is the introduction of new subjects (e.g., new hires) or new
objects (e.g., equipment purchases).  The policy administrator will
then need to deliberate on what changes to the policy must be put in
place, before policy revisions can be implemented. This is a task
commonly known as \Dfn{policy administration}.

In the history of Access Control research, one of the enduring
problems has been to improve the scalability of policy administration.
In other words, access control models are designed to anticipate
changes: when new subjects and objects are introduced over time, it
should not take the policy administrator a lot of deliberation efforts
to revise the policy. In this work, such deliberation overheads are
called the \Dfn{cost of policy administration} (or simply
\Dfn{administration cost}).

For example, instead of having to deliberate about the contents of
every new entries in the access control matrix \cite{Graham1972} when
a new subject or object is created, Role-Based Access Control (RBAC)
promises to reduce policy administration overhead by introducing an
abstraction of subjects known as roles \cite{Sandhu1996}. Permissions
are granted not directly to subjects, but to roles.  When a new
subject is introduced, the policy administrator only needs to decide
which roles the subject shall be assigned to (rather than figuring out
which permissions should be assigned directly to the subject). Since
the number of roles is expected to be much smaller than the number of
subjects and the number of permissions, it is anticipated that the
overall complexity of permission assignment and user assignment is
reduced.  The intuition is that this facilitates policy
administration.

One of the reasons that Attribute-Based Access Control (ABAC)
\cite{Hu2015} has recently attracted the attention of the Access
Control research community is the same promise of making policy
administration scalable, especially in the era of the Internet of
Things (IoT), in which the number of devices grows with the number of
users.  By adopting an intensional style of policy specification
(i.e., specifying the condition of access rather than enumerating the
subjects who should be granted access), ABAC promises to reduce
administration cost when new subjects and objects are introduced.  It
is assumed that the condition of access, if formulated in its most
general form, shall remain the same even when new subjects or objects
are introduced.  Intuitively, this reduces administration cost.

Unfortunately, the savings in policy administration cost in Access
Control research is usually characterized in intuitive terms. There
has been no formal framework to quantify the policy deliberation
efforts required by the policy administrator when new entities (e.g.,
subjects and objects) are created in the protection state. In this
paper, we take the first step to quantify policy administration cost,
so that the benefits of a specific change in policy administration
strategies can be formally accounted for.

We propose to model policy administration in an evolving organization
under the framework of \Dfn{active learning} \cite{Settles2012}. In
active learning, a learner is equipped with a number of queries that
it can use to interrogate a teacher, who possesses complete knowledge
of the target concept. The learner formulates a series of queries
to obtain information about the target concept. With such information
the learner revises and improves its hypothesis of the target concept
over time. Adopting this framework, we model the policy administrator
as the learner. The target concept encapsulated behind the teacher
is the access control matrix of all subjects and objects that can ever
exist. Learner queries correspond to two aspects of reality. First,
some queries allow the learner to discover new entities (i.e.,
subjects and objects). Such queries model organizational evolution.
Second, some other queries correspond to the policy deliberation
efforts of the learner. By asking this second type of query, the
learner discovers the access control characteristics of the new
entities (i.e., who can access what).  The policy administrator
maintains a hypothesis that summarizes what it knows about the
entities.  This hypothesis is a working policy.  As
learning progresses, the policy administrator becomes more and more
informed about the access control characteristics of the entities,
and accordingly improves its policy formulation. The following
summarizes our approach.
\begin{figure}[h]
\centering
\begin{tabular}{rcp{4cm}}
 \textbf{learner} & $\leftrightarrow$ & \emph{policy administrator}\\
 \textbf{target concept} & $\leftrightarrow$ & \emph{access control matrix of all entities that can ever exist}\\
 \textbf{query} & $\leftrightarrow$ & \emph{(a) discovery of new entities or
        (b) deliberation of access control characteristics of entities}\\
 \textbf{hypothesis} & $\leftrightarrow$ & \emph{working policy}
\end{tabular}
\end{figure}

In this modelling approach, the teacher corresponds to
multiple facets of reality: (a) the discovery of new entities and (b)
the deliberation efforts of the policy administrator. By assessing the
\Dfn{query complexity} \cite[Ch.~8]{Kearns1994} of the learning
process, that is, the number of queries required to learn an adequate
hypothesis, we obtain a quantitative characterization of the policy
administration cost incurred to the policy administrator.  With this
framework, we can alter the policy administration strategy (i.e., what
queries to issue) and examine how such alterations impact the query
complexity.

We demonstrate the utility of this approach by applying it to the
administration of \Dfn{protection domains}. The basic idea of
protection domains is that entities (e.g., users) with equivalent
access control characteristics (e.g., needing the same privileges) are
grouped under the same protection domain. Intuitively, this grouping
facilitates policy administration.  Protection domains are almost as
old as the study of Access Control itself and are widely deployed in
our software infrastructure.  An example is the now-classic
\Dfn{domain and type enforcement} \cite{Badger1995}, which has
been implemented in SELinux, which in turn is the foundation of the
Android operating system. Protection domains can also be found in
programming language environments (e.g., Java) and Internet-of-Things
platforms \cite{Carranza2019}.

We do not differentiate subjects and objects, and treat them uniformly
as \Dfn{entities}. As we shall see, this is a generalization rather
than a restriction, as each IoT device plays the roles of both subject
and object simultaneously. We use the term \Dfn{domain-based policy}
to refer to the combination of (a) a set of protection domains, (b) an
assignment of each entity to a protection domain, and (c) a collection
of authorization rules of the form: ``any entity $u$ in protection
domain $d_1$ may exercise access right $r$ over any entity $v$
belonging to protection domain $d_2$.''

Suppose new entities join the organization over time, new entities
with never-before-seen access control characteristics.  Then the
number of protection domains, the assignment of entities to these
domains, and the authorization rules all need to evolve to accommodate
the novelties.  All these incur administration costs to the policy
administrator.  At stake here is the scalability of policy
administration. In an IoT setting, we are talking about tens of
thousands of devices in one organization, the cost of policy
administration could become unmanageable.

As we applied the aforementioned active learning framework to assess
the administration cost for domain-based policies, we noticed a close
analogy between policy evolution and scientific discovery.
Philosophers of science point out that scientists generate new
hypotheses by heuristic reasoning, a process that is inherently
fallible \cite{sep-scientific-discovery, Ippoliti2018}.  In a similar
manner, we found out that heuristics enable the policy administrator
to exploit the conceptualizing instruments (e.g., protection domains,
roles, attributes, relationships) of the underlying access control
model to reduce administration cost.  The price is that the policy
administrator must now commit to fix any detected errors.

We claim the following contributions:
\begin{enumerate}
\item We developed an active learning framework for assessing the
  administration cost involved in revising domain-based policies in an
  evolving organization.  Specifically, we quantified administration
  cost in terms of query complexity (i.e., the number of questions
  that the learner needs to ask).
\item Under this framework, we demonstrated that administration cost
  depends not only on the access control model alone, but also on the
  manner in which policy administration is conducted.  We term the
  latter a policy administration strategy.  We demonstrated that, when
  heuristic reasoning is used in the policy administration strategy,
  using protection domains incurs a lower administration cost than
  when the same policy is represented as an access control matrix.
\item This work suggests a methodology that enables future work to
  study the policy administration cost of an access control model in a
  quantitative manner, and to compare the cost advantages of different
  policy administration strategies.
\end{enumerate}

This paper is organized as follows.  \S \ref{sec:review} formally
introduces domain-based policies, and reviews the theory of
domain-based policies developed in prior work.  Then \S
\ref{sec:framework} introduces an active learning framework for
modelling policy administration, and applies it to study the
administration of domain-based policies.  \S
\ref{sec:tireless-learner} demonstrates that, with a naive policy
administration strategy, domain-based policies offer no cost advantage
over access control matrices.  \S \ref{sec:conservative-learner} then
introduces a heuristic policy administration strategy, which
implements the principle of Occam's Razor.  By allowing the learner to
be occasionally fallible and committing to fix any detected errors,
the overall administration cost can be significantly reduced.  Related
work is surveyed in \S \ref{sec:related-work}, and \S
\ref{sec:conclusion} concludes the paper by presenting the
methodological lessons that future work can draw on.

\section{Domain-Based Policies: A Review}
\label{sec:review}

Our current work is built on the theory of domain-based policies
developed by Zhang and Fong in \cite[\S 2]{DBPM}. We review their
results before proceeding to the presentation of our own
contributions.

\paragraph{Access Control Matrices as Digraphs.} Suppose there is a
fixed set $\Sigma$ of access rights. The members of $\Sigma$ can also
be interpreted as access modes in UNIX, event topics in the IoT
setting, method invocations, etc. An access control
matrix can then be represented as an \Dfn{edge-labelled directed
  graph} (or simply \Dfn{digraph}) $G = (V, E)$, where $V$ is the set
of vertices and $E \subseteq V \times E\times V$ is the set of edges.
Each vertex represents an entity such as a subject, an object, or a
device in the IoT setting. An edge $(u, a, v) \in E$ represents the
permission for entity $u$ to exercise access right $a$ over entity
$v$.  Essentially, a digraph exhaustively enumerates the permissions
of the corresponding access control matrix in the form of edges.  We
also write $V(G)$ and $E(G)$ for $V$ and $E$ respectively.  Common
graph-theoretic concepts such as subgraphs, isomorphism, etc, can be
defined as usual. Given $U \subseteq V(G)$, we write $G[U]$ for the
\Dfn{subgraph of $G$ induced by $U$}. Here, the vertex set of $G[U]$
is simply $U$, and for $u,v \in U$ and $a \in \Sigma$, $(u,a,v)$ is an
edge in $G[U]$ iff $(u,a,v) \in E(G)$.

\paragraph{Domain-Based Policies.} Given a digraph $G$, a
\Dfn{domain-based policy} is a pair $(H, \pi)$, where $H$ is a digraph
and $\pi : V(G) \rightarrow V(H)$ maps vertices of $G$ to vertices of
$H$. The intention is that the members of $V(H)$ are the
\Dfn{protection domains} (or simply \Dfn{domains}). The mapping $\pi$
assigns every entity in $G$ to a domain.  When an access request
$(u,a,v)$ is received by the protection mechanism, the request is
granted iff $(\pi(u), a, \pi(v)) \in E(H)$. In other words, an edge
$(x, a, y)$ in $H$ signifies that any entity belonging to domain $x$
may exercise access right $a$ over any entity belonging to domain
$y$. Conversely, absence of an edge signifies the denial of access.
Typically, we want $\pi$ to map entities with equivalent access
control characteristics to the same domain.

\paragraph{Correct Enforcement.} Given an authorization request
$(u,a,v)$, a poorly formulated domain-based policy $(H, \pi)$ for
$G$ may produce a different authorization decision than $G$ itself.
We say that $(H, \pi)$ \Dfn{enforces} $G$ whenever $(u,a,v) \in E(G)$
iff $(\pi(u), a, \pi(v)$ for every $u,v \in V(G)$ and $a \in \Sigma$.

A function $\pi: V(G) \rightarrow V(H)$ is a \Dfn{strong homomorphism}
from $G$ to $H$ whenever $(u,a,v) \in E(G)$ iff $(\pi(u), a, \pi(v)$.
Therefore, domain-based policy $(H, \pi)$ enforces $G$ iff
$\pi$ is a strong homomorphism from $G$ to $H$.

\paragraph{Digraph Summary.} When $(H, \pi)$ enforces $G$, $H$
properly summarizes the authorization decisions using domains as an
abstraction of entities. In theory, $G$ is always a ``summary'' of
itself, but not a very succinct one. We desire to compress the
information in $G$ as much as possible by grouping as many entities
into the same domain as possible. In other words, we desire the most
succinct summary of $G$. Digraph $H$ is a \Dfn{summary} of digraph $G$
iff (a) $G$ is strongly homomorphic to $H$, and (b) $G$ is not
strongly homomorphic to any proper subgraph of $H$.

Suppose $H$ is a summary of $G$ through the strong homomorphism
$\pi: V(G) \rightarrow H(G)$.  Then $H$ and $\pi$ have three important
characteristics.  First, $\pi$ must be a surjective function (meaning
a summary has no redundant vertices).  Second, $H$ is
\Dfn{irreducible}, meaning that every summary of $H$ is isomorphic to
$H$ itself.  In other words, a summary cannot be further
summarized.\footnote{This notion of minimality does not apply to
  infinite digraphs. One can construct an infinite series of infinite
  digraphs $G_1$, $G_2$, $\ldots$, so that for $i \in \mathbb{N}$,
  $G_{i+1}$ is a proper subgraph of $G_i$ and $G_i$ is strongly
  homomorphic to $G_{i+1}$.  See \cite{Hell1992} for examples of such
  a series. Therefore, when the notion of summary is invoked in this
  paper, it is always concerned with the summary of a finite digraph,
  even though the latter could be a subgraph of an infinite
  digraph.  } Third, every summary of $G$ is isomorphic to $H$ (meaning
a summary is unique up to isomorphism).

\paragraph{Summary Construction.} Zhang and Fong devised a tractable
means for constructing the summary of a given digraph $G$. Their
method is based on an equivalence relation $\equiv_G$ over the vertex
set $V(G)$.  In particular, we write $u \equiv_G v$ (meaning $u$ is
\Dfn{indistinguishable} from $v$) iff both conditions below hold for
every $a \in \Sigma$:
\begin{enumerate}
\item The four edges, $(u,a,u)$, $(u,a,v)$, $(v,a,u)$, and $(v,a,v)$
  either \emph{all} belong to $E(G)$ or \emph{all} does not
  belong to $E(G)$.
\item For every $x \in V(G) \setminus \{u,v\}$,
  \begin{enumerate}
  \item $(u,a,x) \in E(G)$ iff $(v,a,x) \in E(G)$, and
  \item $(x,a,u) \in E(G)$ iff $(x,a,v) \in E(G)$.
  \end{enumerate}
\end{enumerate}
In other words, $u$ and $v$ are indistinguishable iff their adjacency
with other vertices are identical.  We also write
$\mathit{adj}_G(u, v)$ to denote the set
$\{\, +a \,\mid\, (u,a,v) \in E(G) \,\} \cup \{\, -a \,\mid\, (v, a,
u) \in E(G) \,\}$.  Thus $x \equiv_G y$ iff
$\mathit{adj}_G (x, z) = \mathit{adj}_G (y, z)$ for every
$z \in V(G)$, iff $\mathit{adj}_G (z, x) = \mathit{adj}_G (z, y)$ for
every $z \in V(G)$.

Exploiting the fact that the indistinguishability of two given
vertices can be tested in linear time, Zhang and Fong devised an
algorithm, \textsc{Summarize}, which takes as input a digraph $G$, and
produces a domain-based policy $(H, \pi)$, so that $H$ is both a
summary and a subgraph of $G$, and $\pi$ is the corresponding
surjective strong homomorphism.  The algorithm runs in $O(kn^3)$ time,
where $n = |V(G)|$ and $k = |\Sigma|$.

\section{Policy Administration as Active Learning} 
\label{sec:framework}

In an evolving organization, we do not know of all the entities that
will ever join the organization.  As the organization grows and
technology advances, new entities will be created. These entities may
have access requirements and characteristics that are radically
different from the existing ones.  It is simply unrealistic to expect
that the domain-based policies we constructed using Zhang and Fong's
\textsc{Summarize} algorithm \cite[\S 2]{DBPM} will continue to work
in the future as new entities join the mix. The policy administrator
will have to assign the new entities to existing protection domains or
even formulate new protection domains. Our goal in this section is to
create a formal model for this ongoing policy administration process,
so that we can quantify the cost of policy administration.  One way to
think about this is that the access control matrix $G$ evolves over
time as more and more vertices join the digraph.  Yet a more fruitful
way to capture this dynamism in a formal model is to envision a
countably infinite digraph $G$, complete with all the vertices that
will ever join the organization, but the knowledge of this infinite
graph is incrementally disclosed to the policy administrator.  The
administrator's task is to grow her understanding of $G$ over time,
revising her summary $H$ so that $H$ enforces a larger and larger
induced subgraph of $G$.  To formalize this dynamism of policy
administration, we adopt an active learning framework
\cite{Settles2012}, one that is inspired by the work of Angluin
\cite{Angluin1987} from the literature of computational learning
theory \cite{Kearns1994}.

We introduce some notations before we describe our active learning
protocol. 
\begin{definition}[Error]
  Suppose $(H, \pi)$ is a domain-based policy for digraph $G$.  A
  \Dfn{grant error} is a request $(u, a, v)$ such that
  $(\pi(u), a, \pi(v)) \in E(H)$ but $(u, a, v) \not\in E(G)$.  A
  \Dfn{deny error} is a request $(u, a, v)$ such that
  $(\pi(u), a, \pi(v)) \not\in E(H)$ but $(u, a, v) \in E(G)$.  An
  \Dfn{error} $(u, a, v)$ is either a grant error or a
  deny error.  Let $\mathcal{E}(G, H, \pi)$ denote the set of all
  errors.
\end{definition}

Our active learning protocol involves two parties, the \Dfn{learner}
and the \Dfn{teacher}.  Loosely speaking, the goal of the learner, who
is a reactive process, is to gradually discover the structure of a
countably infinite digraph $G$. This graph is encapsulated behind a
hypothetical teacher.  Initially, the learner has no information about
$G$.  The learner acquires information about $G$ by issuing
\Dfn{queries} to the teacher. The teacher is assumed to be truthful:
it never lies about $G$.  The protocol supports three queries:

\begin{asparaenum}
\item \textbf{Next Vertex Query (NVQ):} When the query
  $\textsc{NVQ}()$ is issued by the learner, the teacher will return a
  never-before-seen vertex from $G$.  This query models the
  recruitment of a new user or the acquisition of a new resource by
  the organization.  Let $U$ be the (finite) set of all vertices that
  the teacher has returned so far through NVQ.  It is assumed that the
  teacher tracks the contents of $U$.  The teacher may return vertices
  of $G$ in any possible order.
\item \textbf{Connection Query (CNQ):} The learner issues
  $\textsc{CNQ}(u, a, v)$ to inquire about the existence of the edge
  $(u, a, v)$ in $G$.  Here, $u,v \in U$ and $a \in \Sigma$.  The
  teacher returns a boolean value.  The CNQ query is intended to model
  the cognitive overhead incurred by the policy administrator when the
  latter deliberates on whether to allow entity $u$ to perform
  operation $a$ against entity $v$.
\item \textbf{Hypothesis Testing Query (HTQ):} The learner invokes the
  query $\textsc{HTQ}(H, \pi)$, where $H$ is a finite digraph and
  $\pi : U \rightarrow V(H)$ is a function, to check if $H$ and $\pi$
  properly summarize the accessibility encoded in the induced subgraph
  $G[U]$.  The teacher responds by returning the set
  $\mathcal{E}(G[U], H, \pi)$ of errors.  This query models the
  releasing of the domain-based policy $(H, \pi)$.  Experiences with
  $(H, \pi)$ are gained and errors are identified.\footnote{ This
    practice of deploying a ``good enough'' policy that may still
    contain errors is corroborated by the findings of He \emph{et al.}
    \cite{He2018}, in which they found that users of smart home
    devices indeed tolerate the existence of both grant errors and
    deny errors in their policy formulation when they are
    still learning about the effects of adopting a certain access
    control policy.  }  The error set represents knowledge about the
  policy $(H, \pi)$ that is acquired outside of policy deliberation.
  Such knowledge may come from stakeholder feedback, expert scrutiny,
  or empirical experiences obtained through the deployment of the
  policy.  Depending on the application domain, this knowledge may
  also come from a combination of the above sources.  Note that the
  error set $\mathcal{E}(G[U], H, \pi)$ concerns only the finite
  subgraph $G[U]$ induced by the set $U$ of previously returned
  vertices. The learner is not supposed to know anything about the
  rest of $G$.
\end{asparaenum}

\begin{comment}
In summary, the teacher is an abstraction of several organizational
realities: (a) the introduction of new entities over time, (b) the
deliberation of the policy administrator, and (c) assessment of the
policy obtained from external sources.
\end{comment}

Given the queries above, the intention is for the learner to
strategize the questioning in such a way that it eventually learns a
domain-based policy $(H, \pi)$ for $G$.  The criteria of successful
learning involve two aspects. The first criterion concerns the quality
of $H$ and $\pi$. That is, $H$ should be a summary of what the learner
knows about $G$.  The second criterion concerns how fast this learning
process converges.  We capture these two criteria in the following
definition:
\begin{definition} \label{def:successful}
  A learner is \Dfn{successful} iff it satisfies the
  two criteria below:
  \begin{description}
  \item[SC-1.] When HTQ is invoked, the argument $H$ must be
    irreducible and the argument $\pi$ must be surjective.
  \item[SC-2.] Once an NVQ query has been issued, the learner must
    issue at least one HTQ that returns an empty set of errors, before
    the next NVQ can be issued.
  \end{description}
\end{definition}
Success criterion \textbf{SC-1} is inspired by the fact that if $H$ is
a summary of $G[U]$ via strong homomorphism $\pi$, then $\pi$ must be
surjective and $H$ must be irreducible (see \S \ref{sec:review}) This
success criterion requires the learner to at least attempt to
construct a summary $H$ of $G[U]$.  Success criterion \textbf{SC-2} is
an aggressive learning schedule: the learner must fix all errors
before progressing to consider another new vertex of $G$.  These two
success criteria are by no means the only ones possible.  We plan to
explore the implications of alternative criteria in future work.

A successful learner is the computational model of a \Dfn{policy
  administration strategy}. While such a strategy is presented
algorithmically, it is not intended to be program code executed by a
computer. Instead, the strategy prescribes how the policy
administrator (\emph{a human}) shall respond to the introduction of new
entities: e.g., what policy deliberation efforts shall be conducted
(CNQ), when to assess the revised policy (HTQ), and how to fix up a
policy when errors are discovered.  We are interested in assessing the
performance of successful learners (policy administration strategies).
What concerns us is not so much time complexity: we consider the
learner acceptable so long as the computational overhead between
successive queries is a polynomial of $|U|$.  In active learning
\cite[Ch.~8]{Kearns1994}, the competence of a learner is evaluated by
its \Dfn{query complexity}, that is, the number of queries issued by
the learner.  We adapt this practice as follows.
\begin{itemize}
\item The three queries (\textsc{NVQ}, \textsc{CNQ}, and \textsc{HTQ})
  are intended to model different aspects of reality. We do not count
  them in the same way.
\item The learner is a reactive process (it never terminates) because
  the digraph to be learned is infinite. Because of this, the number
  of queries issued by the learner may grow to infinite as well. To
  cope with this, \textbf{SC-2} demands that learning occurs in
  rounds. Every round begins with the invocation of an \textsc{NVQ}.
  After that some number of \textsc{CNQ}s and \textsc{HTQ}s
  follow. The round ends with an \textsc{HTQ} that returns an empty
  set of errors.  We therefore use the number of rounds
  (i.e., the number of \textsc{NVQ}s) as an ``input parameter,'' and
  express the number of other queries (or errors) as a
  function of this parameter.
\item Policy administration overhead is captured by \textsc{CNQ}s.  We
  therefore quantify administration cost as the number of
  \textsc{CNQ}s issued when $n$ rounds of learning have occurred
  (i.e., $n$ invocations of \textsc{NVQ}s have been issued so far).
\item As for \textsc{HTQ}s, we are concerned about the total number of
  errors committed in $n$ rounds of learning rather than the number of
  \textsc{HTQ} invocations.
\end{itemize}

\section{Tireless Learner}
\label{sec:tireless-learner}

\begin{algorithm}[t] 
  \caption{The Tireless Learner. \label{algo:tireless}}
%  \SetInd{.1em}{.2em}
  Initialize digraph $G^*$ so that $V(G^*) = E(G^*) = \emptyset$\;
  \While{$\mathit{true}$}{
    $u = \textsc{NVQ}()$\label{line:eager-NVQ}\;
    $V(G^*) = V(G^*) \cup \{ u \}$\;
    \ForEach{$a \in \Sigma$\label{line:eager-CNQ-begin}}{
      \ForEach{$v \in V(G^*)$}{
        \lIf{$\textsc{CNQ}(u,a,v)$}{$E(G^*) =
          E(G^*) \cup \{ (u,a,v)\}$}
      }
      \ForEach{$v \in V(G^*) \setminus \{u\}$}{
        \lIf{$\textsc{CNQ}(v,a,u)$}{$E(G^*) =
          E(G^*) \cup \{ (v,a,u)\}$\label{line:eager-CNQ-end}}
      }
    }
    $(H, \pi) = \textsc{Summarize}(G^*)$\label{line:eager-offline}\;
    $\mathcal{E} = \textsc{HTQ}(H, \pi)$\label{line:eager-HTQ}%\;
    \tcp*{$\mathcal{E}$ is always $\emptyset$.}
  }
\end{algorithm}

To demonstrate how the learning protocol works, we explore here a
naive learner: the \Dfn{Tireless Learner} (Algorithm
\ref{algo:tireless}). The Tireless Learner captures the following
policy administration strategy: As each new entity $u$ is revealed,
the policy administrator deliberates on the contents of every new
entry in the access control matrix $G[U]$, and then summarizes the
updated $G[U]$ into a domain-based policy $(H, \pi)$.  A number of
technical observations can be made about the Tireless Learner:
\begin{itemize}
\item An invariant of the outermost while-loop is that $G^* = G[U]$,
  where $U$ is the (finite) set of vertices that has been returned so
  far by the \textsc{NVQ}.
\item When a new vertex is returned through the NVQ (line
  \ref{line:eager-NVQ}), the Tireless Learner invokes \textsc{CNQ}s
  exhaustively to discover how the new vertex is connected to the rest
  of $G^*$ (lines
  \ref{line:eager-CNQ-begin}--\ref{line:eager-CNQ-end}).  This is how
  the learner maintains the invariant $G^* = G[U]$.
\item \textsc{Summarize} is invoked on line
  \ref{line:eager-offline} to compute a domain-based policy
  $(H, \pi)$ of $G[U]$.
\item Since $H$ is a summary of $G[U]$ via the strong homomorphism
  $\pi$, $H$ is irreducible and $\pi$ is surjective (see \S
  \ref{sec:review}).  Thus \textbf{SC-1} is satisfied.
\item Since $(H, \pi)$ enforces $G[U]$ (see \S \ref{sec:review}), the
  set $\mathcal{E}$ of errors returned by \textsc{HTQ} on line
  \ref{line:eager-HTQ} is always an empty set. Thus \textbf{SC-2} is
  satisfied.
\item The Tireless Learner is successful.
\end{itemize}

We say that the Tireless Learner is naive because it issues CNQs
relentlessly. The administration cost is therefore maximized.  We
quantify the administration cost in the following theorem.
\begin{theorem}[Administration Cost]
  \label{thm:tireless-admin-cost}
  Let $k$ be $|\Sigma|$ and $n$ be the number of \textsc{NVQ}s issued
  by the Tireless Learner so far.  Then the \textsc{CNQ} has been
  invoked $kn^2$ times.
\end{theorem}
\begin{proof}
  During the $i$th iteration of the while-loop, the learner invokes
  $k(2i-1)$ \textsc{CNQ}s to update $G^*$ (lines
  \ref{line:eager-CNQ-begin}--\ref{line:eager-CNQ-end}).  After $n$
  iterations, the total number of \textsc{CNQ}s is
  $k\times (1 + 3 + 5 + \cdots + (2n-1)) = kn^2$.
\end{proof}
This means the Tireless Learner, as a policy administration strategy,
can successfully learn a domain-based policy by incurring an
administration cost of $kn^2$. Note that $kn^2$ is exactly the number
of bits of information carried by an access control matrix. In other
words, the Tireless Learner deliberates exhaustively on every bit of
information in the access control matrix.  One would have achieved
this administrative cost ($kn^2$) simply by tracking an access control
matrix. \emph{Even though protection domains are used, the Tireless
  Learner did not take advantage of this access control abstraction to
  reduce its administration cost.}  This observation anticipates a key
insight offered by this work: \emph{The merit of an access control
  model lies not only in the model itself. The model is able to scale
  with the growing number of entities because it is accompanied by a
  complementary policy administration strategy that exploits the
  conceptualizing instruments (e.g., protection domains, roles,
  attributes, relationships) offered by the model.}  An alternative
policy administration strategy for domain-based policy will be
presented in the \S \ref{sec:conservative-learner}.  As we consider
alternative policy administration strategies, $kn^2$ will be the
baseline of comparison. The goal is to do better than tracking only an
access control matrix, so that the administration cost does not grow
quadratically with the number of entities.

\begin{comment}
That the administration cost (i.e., the number of \textsc{CNQ}s) grows
quadratically with the number of devices acquired by the organization
(i.e., the number of NVQs) is unacceptable, as an organization could
be working with a large number of devices.  This quadratic growth of
administration cost is the result of the Tireless Learner insisting on
producing no errors.

There is actually a trade-off between administration cost and errors.
The errors returned by the \textsc{HTQ} convey information about the
adjacency of vertices.  If the domain-based policy $(H, \pi)$ allows
the access request $(u,a,v)$, but $(u,a,v)$ is actually an error, then
we know that $(u,a,v)$ is not an edge in $G[U]$.  In fact, even when
the \textsc{HTQ} confirms that $(u,a,v)$ is not an error, we still
learn something new: the authorization decision of $(H, \pi)$
regarding request $(u,a,v)$ is the same as that of the extensionally
specified policy $G[U]$.  Consequently, the \textsc{HTQ} represents an
alternative to the \textsc{CNQ}: adjacency information can be obtained
either by issuing \textsc{CNQ}s (policy deliberation) or by invoking
\textsc{HTQ}s (policy assessment). If we are willing to tolerate
a small number of errors, we can reduce the administration cost
significantly, as we shall see in the next section.
\end{comment}

\section{Conservative Learner}
\label{sec:conservative-learner}

A policy administration strategy (i.e., a learner) can lower
administration cost by performing heuristic reasoning. Rather than
exhaustively deliberating on every bit of information in the access
control matrix, the learner can make use of a ``fallible'' learning
strategy to reduce the deliberation overhead. In exchange, errors may
be produced, and the policy needs to be fixed when errors are
detected. The use of heuristic strategies is a common phenomenon in
scientific discovery \cite{sep-scientific-discovery}. When a scientist
generates candidate hypotheses, heuristics may guide the process
\cite{Ippoliti2018}.  And heuristics are by definition not
error-proofed.  In a similar vein, the policy administrator may engage
in fallible, heuristic reasoning when it constructs a policy.  In
fact, there is empirical evidence that such a trade-off between the
efficiency of policy deliberation and the correctness of policy
formulation indeed occurs in the context of IoT systems, when the
timely deployment of policies is desired \cite{He2018}.

\begin{algorithm}[t]
  \caption{The Conservative Learner.\label{algo:conservative-learner}}
%  \SetInd{.2em}{.4em}
  $u = \textsc{NVQ}()$\label{line:cons-initial-NVQ}\;
  let digraph $H$ be $(V_0,E_0)$,
  where $V_0 = \{u\}$, $E_0 = \emptyset$\label{line:cons-init-vertices}\;
  \ForEach{$a \in \Sigma$}{
    \lIf{$\textsc{CNQ}(u,a,u)$}{
      $E(H) = E(H) \cup
      \{(u,a,u)\}$\label{line:cons-init-edges}}
  }
  let function $\pi = \{ u \mapsto u\}$\label{line:cons-init-pi}\;
  let decision tree $\mathcal{T}$ be a leaf $n$,
  with $\ell(n) = u$\label{line:cons-init-tree}\;
  $\mathcal{E} = \textsc{HTQ}(H,\, \pi)$\label{line:cons-initial-confirm}\;
  \While{true\label{line:cons-main-loop}}{
    $u = \textsc{NVQ}()$\label{line:cons-NVQ}\;
    $w = \textsc{Classify}(\mathcal{T},\, u)$\label{line:cons-classify}\;
    $\pi' = \pi[u \mapsto w]$\label{line:cons-extend-pi}\;
    $\mathcal{E} = \textsc{HTQ}(H,\, \pi')$\label{line:cons-first-confirm}\;
    \If{$\mathcal{E} = \emptyset$\label{line:cons-detect-error}}{
      $\pi = \pi'$\label{line:cons-no-revision}\;
    }
    \Else{
      $(\mathcal{T},\pi) = \textsc{Revise}(\mathcal{T}, H, \pi,
      u, w,\mathcal{E})$\label{line:cons-revise-tree}\;
      $V = \Ran{\pi}$\label{line:cons-recompute-V-of-H}\;
      $E = \{ (u,a,v)
      \in V \times\Sigma\times V
      \mid
\textsc{Edg}(u,a,v,H,\pi',\mathcal{E})\}$\label{line:cons-recompute-E-of-H}\;
      $H = (V,E)$\label{line:cons-revise-H}\;
      $\mathcal{E} =
      \textsc{HTQ}(H,\,\pi)$\label{line:cons-confirm-revision}%\;
      \tcp*{$\mathcal{E}$ is always $\emptyset$.}
    }
  }
\end{algorithm}

This section presents such a learner.  The design of this learner is
based on the well-known principle of \Dfn{Occam's Razor}
%\cite[Ch.~6]{Anthony1992}
\cite[Ch.~2]{Kearns1994}: the learner
strives to reuse the simple summary that it has learned so far, until
external feedback forces it to abandon the existing summary for a more
complex one. Operationally, it means that the learner always assumes
that the new vertex returned by the teacher is indistinguishable from
some previously seen vertex, until errors prove that they are in fact
distinguishable.

Why would this presumption of indistinguishability reduce
administration cost?  While the number of entities (vertices in
digraph $G$) may be infinite, the number of protection domains
(vertices in the summary $H$) is relatively small. Once the learner
has seen a sample entity in an equivalence class, all the future
entities of the same equivalence class look the same: they share the
same adjacency pattern as the sample. After the learner has learned
all the equivalence classes, no new adjacency patterns need to be
learned.  The remaining learning process is simply a matter of
classifying new entities into one of the known equivalence classes. As
we shall see in \S \ref{sec:conservative-complexity}, this latter task
of classification requires only a number of \textsc{CNQ}s that is a
function of the number of protection domains rather than the number of
vertices seen.  The administration cost is therefore reduced
significantly.

This new learner is called the \Dfn{Conservative Learner} (Algorithm
\ref{algo:conservative-learner}).  Here we outline the high-level
ideas, and leave the details to the rest of the section.
\begin{asparaenum}
\item In the beginning of each round, the teacher returns a new vertex
  $u$ via an \textsc{NVQ} (line \ref{line:cons-NVQ}). Rather than
  asking \textsc{CNQ}s exhaustively to uncover the adjacency between
  $u$ and the existing vertices, the learner acts
  ``\emph{conservatively}\/'': It assumes that $u$ is
  indistinguishable from some existing vertex.
\item It uses a classifier to classify $u$ into one of the known
  equivalence classes (line \ref{line:cons-classify}).  That
  classifier is a decision tree $\mathcal{T}$. The decision nodes of
  $\mathcal{T}$ correspond to \textsc{CNQ}s that must be invoked in
  order to obtain a classification. Since the number of equivalence
  classes is assumed to be small, $\mathcal{T}$ is small, and thus the
  number of \textsc{CNQ}s required to classify $u$ is significantly
  smaller than the exhaustive discovery of adjacency.
\item The classification result allows the learner to extend $\pi$ by
  assigning $u$ to an existing protection domain (line
  \ref{line:cons-extend-pi}).  (The notation $\pi[u \mapsto w]$
  denotes a function $f$ such that $f(x) = w$ if $x = u$, and
  $f(x) = \pi(x)$ otherwise.)  $H$ remains the same.
\item Of course, the assumption that the new vertex $u$ is
  indistinguishable from a previously seen vertex may or may not be
  true. That is why the learner employs the \textsc{HTQ} to confirm this
  (line \ref{line:cons-first-confirm}). If no errors are
  returned, then the bet pays off (line \ref{line:cons-no-revision}).
  The premise is that, after enough equivalence classes have been
  discovered, this case is the dominant case.
\item If the teacher returns actual errors, then
  the decision tree $\mathcal{T}$
  and the working policy $(H, \pi)$ are revised
  to eliminate the errors (lines
  \ref{line:cons-revise-tree}--\ref{line:cons-confirm-revision}).
\end{asparaenum}

A detailed exposition of Algorithm \ref{algo:conservative-learner} is
given below. First, we introduce decision trees (\S
\ref{sec:conservative-trees}). We then examine how equivalence classes
evolve as new vertices are revealed by the teacher (\S
\ref{sec:conservative-equivalence}). This prepares us to understand
the revision of the decision tree and the working policy (\S
\ref{sec:tree-revision}). Lastly, we assess the correctness (\S
\ref{sec:conservative-invariant}) and administration cost (\S
\ref{sec:conservative-complexity}) of the Conservative Learner.

\begin{algorithm}[t]
  \caption{\textsc{Classify}($\mathcal{T}$, $u$)\label{Algo:Classify}}
  \lIf{$\mathcal{T}$ is a leaf}{\Return{$\ell(\mathcal{T})$}}
  \Switch{$\mathit{test}(\mathcal{T})$}{
    \Case{$\mathit{loop}(a)$}{
      \If{\textsc{CNQ}($u$, $a$, $u$)}{
        \Return
        \textsc{Classify}($\mathit{left}(\mathcal{T})$, $u$)\;
      }
      \lElse{\Return
        \textsc{Classify}($\mathit{right}(\mathcal{T})$, $u$)}
    }     
    \Case{$\mathit{to}(a, v)$}{
      \If{\textsc{CNQ}($u$, $a$, $v$)}{
        \Return
        \textsc{Classify}($\mathit{left}(\mathcal{T})$, $u$)\;
      }
      \lElse{\Return
        \textsc{Classify}($\mathit{right}(\mathcal{T})$, $u$)}
    }
    \Case{$\mathit{from}(v, a)$}{
      \If{\textsc{CNQ}($v$, $a$, $u$)}{
        \Return
        \textsc{Classify}($\mathit{left}(\mathcal{T})$, $u$)\;
      }
      \lElse{\Return
        \textsc{Classify}($\mathit{right}(\mathcal{T})$, $u$)}
    }     
  }
\end{algorithm}

\begin{algorithm}[t]
  \caption{\textsc{Edg}($u$, $a$, $v$,
    $H$, $\pi'$, $\mathcal{E}$)\label{algo:edge}}
  \Return
  $\,(\pi'(u),a,\pi'(v))
  \in E(H)
  \,\,\,\operatorname{xor}\,\,\,
  (u,a,v)
  \in \mathcal{E}$\;
\end{algorithm}

\subsection{Decision Trees}
\label{sec:conservative-trees}

The Conservative Learner presumes that a new vertex $u$ returned by
the teacher is indistinguishable from an already seen vertex.  The
learner then employs a decision tree $\mathcal{T}$ to classify $u$ to
an existing protection domain, hoping that the summary $H$ does not
need to be revised.  In short, a decision tree captures the heuristic
knowledge of the Conservative Learner.  We introduce the structure and
semantics of a decision tree in the following.

\begin{definition}[Decision Tree]
  Suppose $G$ is a digraph. A \Dfn{decision tree $\mathcal{T}$ (for
    $G$)} is a finite binary tree defined as follows:
  \begin{itemize}
  \item A decision tree $\mathcal{T}$ is either a \Dfn{leaf} or
    a \Dfn{decision node}.
  \item If $\mathcal{T}$ is a leaf, then it has a \Dfn{label}
    $\ell(\mathcal{T})$, which is a vertex in $G$.
  \item If $\mathcal{T}$ is a decision node,
    then it has a \Dfn{test} $\mathit{test}(\mathcal{T})$,
    a \Dfn{left subtree} $\mathit{left}(\mathcal{T})$,
    and a \Dfn{right subtree} $\mathit{right}(\mathcal{T})$.
    Both $\mathit{left}(\mathcal{T})$ and
    $\mathit{right}(\mathcal{T})$ are decision trees.
    The test $\mathit{test}(\mathcal{T})$ has one of the
    following three forms:
    \begin{itemize}
    \item $\mathit{loop}(a)$, where $a \in \Sigma$,
    \item $\mathit{to}(a, v)$, where $a \in \Sigma$ and
      $v \in V(G)$, or
    \item $\mathit{from}(v, a)$, where $v \in V(G)$ and $a \in
      \Sigma$.
    \end{itemize}
    Intuitively, $\mathit{test}(\mathcal{T})$ prescribes a test to be
    performed, and the left and right subtree represent respectively
    the ``yes''-branch and ``no''-branch of the test.
  \end{itemize}
\end{definition}
A decision tree $\mathcal{T}$ can be used for classifying vertices
from $G$.  Specifically, Algorithm \ref{Algo:Classify} specifies the
semantics of decision trees: the algorithm classifies a vertex $u$ of
$G$ as one of the leaf labels of $\mathcal{T}$. The process involves
invoking \textsc{CNQ}s.

The intention is that each leaf of $\mathcal{T}$ corresponds to an
equivalence class induced by $\equiv_G$. If a leaf $n$ corresponds to
an equivalence class $\mathcal{C}$ of $\equiv_G$, then the label
$\ell(n)$ is a member of $\mathcal{C}$. This vertex $\ell(n)$ is known
as the \Dfn{representative} of $\mathcal{C}$. In short, a decision
tree $\mathcal{T}$ classifies a vertex $u$ of $G$ to the
representative of the equivalence class to which $u$ belongs.

\subsection{Evolution of Equivalence Classes}
\label{sec:conservative-equivalence}

The Conservative Learner tracks a summary $H$ of $G[U]$, where $U$ is
the (finite) set of vertices returned so far by the teacher.  Each
vertex of $H$ is essentially a representative of an equivalence class
induced by $\equiv_{G[U]}$.  Here we make the following inquiry:
\emph{As the teacher returns more and more vertices (i.e., as $U$
  becomes bigger and bigger), how will the equivalence classes change
  accordingly?}  Answers to this question will help us better
understand the process by which decision trees and summary graphs are
revised (lines
\ref{line:cons-revise-tree}--\ref{line:cons-revise-H}).

The first observation is that distinguishable vertices remain
distinguishable as more and more vertices are revealed by the
teacher. 
\begin{proposition} \label{prop:dist_vertices}
  Suppose $G$ is a digraph and $U_1 \subseteq U_2 \subseteq V(G)$.
  Let $G_1$ be $G[U_1]$, $G_2$ be $G[U_2]$, $\equiv_1$ be
  $\equiv_{G_1}$, and $\equiv_2$ be $\equiv_{G_2}$.  Then for
  $x,y \in U_1$, $x \not\equiv_1 y$ implies $x \not\equiv_2 y$.
\end{proposition}
\begin{proof}
  We prove the contrapositive: $x \equiv_2 y$ implies $x \equiv_1
  y$. Note that $G_2$ contains all the vertices of $G_1$.  Thus,
  according to the definition of $\equiv_G$ in \S \ref{sec:review},
  the requirements of indistinguishability is stronger in $G_2$ than
  in $G_1$.
\end{proof}
Once two vertices are found to be distinguishable, they remain so
throughout the rest of the learning process. In other words,
equivalence classes do not ``merge with one another'' or ``bleed into
one another.''

The revelation of new vertices may cause two previously
indistinguishable vertices to become distinguishable.  This occurs
only when the new vertex contains genuinely new structural information
about $G$.  Otherwise, equivalence classes remain the same.  This
observation is formalized in the following proposition.
\begin{proposition} \label{prop:indist_vertices}
  Suppose $G$ is a digraph, $U_1 \subseteq V(G)$,
  $u \in V(G) \setminus U_1$, and $U_2 = U_1 \cup \{u\}$.  Let $G_1$
  be $G[U_1]$, $G_2$ be $G[U_2]$, $\equiv_1$ be $\equiv_{G_1}$, and
  $\equiv_2$ be $\equiv_{G_2}$.  Suppose further that
  $u \equiv_2 v$ for some $v \in U_1$.  Then for every $x, y \in U_1$,
  $x \equiv_1 y$ implies $x \equiv_2 y$.
\end{proposition}
Recall the definition and properties of the notation
$\mathit{adj}_G(x,y)$ in \S \ref{sec:review} as they are used heavily
in the following proof.
\begin{proof}
  Assume $x \equiv_1 y$, we show that $x \equiv_2 y$.  To that end,
  consider a vertex $z$ in $G_2$. We show that
  $\mathit{adj}_{G_2}(x, z) = \mathit{adj}_{G_2}(y, z)$.  There are
  two cases.
  \begin{enumerate}
  \item \textbf{Case $z \in U_1$}: The adjacency among vertices in
    $G_1$ remains the same in $G_2$.  Since $x \equiv_1 y$, we know
    that
    $\mathit{adj}_{G_2}(x, z) = \mathit{adj}_{G_1}(x,z) =
    \mathit{adj}_{G_1}(y, z) = \mathit{adj}_{G_2} (y,z)$.
  \item \textbf{Case $z = u$}: Since $u \equiv_2 v$, we have
    $\mathit{adj}_{G_2} (x,u) = \mathit{adj}_{G_2} (x,v)$ and
    $\mathit{adj}_{G_2} (y, u) = \mathit{adj}_{G_2}(y,v)$.  But $x$,
    $y$, and $v$ are all vertices from $G_1$, adjacency among them
    remains the same in $G_2$, and thus
    $\mathit{adj}_{G_2}(x,v) = \mathit{adj}_{G_1}(x,v)$ and
    $\mathit{adj}_{G_2}(y,v) = \mathit{adj}_{G_1}(y,v)$.  Since
    $x \equiv_1 y$,
    $\mathit{adj}_{G_1}(x,v) = \mathit{adj}_{G_1}(y,v)$.  Therefore,
    $\mathit{adj}_{G_2}(x,u) = \mathit{adj}_{G_2}(y, u)$.
  \end{enumerate}
  In other words, $\mathit{adj}_{G_2}(x,z) = \mathit{adj}_{G_2}(y,z)$
  for arbitrary $z \in V(G_2)$.  We thus conclude that $x \equiv_2 y$.
\end{proof}

A number of implications follow from Proposition
\ref{prop:indist_vertices}:
\begin{enumerate}
\item When the teacher returns a new vertex $u$ that is
  indistinguishable from a previously seen vertex $v$, the equivalence
  classes do not change (except for $u$ to join the equivalence class
  of $v$).  We shall see that this is the dominant case (\S
  \ref{sec:conservative-complexity}).
\item Otherwise, the new vertex $u$ is distinguishable from every
  other known vertex, and thus $u$ belongs to a new equivalence class
  for which it is the only member.  We call $u$ a \Dfn{novel vertex}.
\item The revelation of a novel vertex could cause previously
  indistinguishable vertices to become distinguishable. By Proposition
  \ref{prop:dist_vertices}, such changes take the form of
  \emph{splitting} an existing equivalence class into multiple
  equivalence classes. This will explain why we later on perform
  ``\Dfn{splitting}'' when we revise a decision tree (\S
  \ref{sec:tree-revision}).
\end{enumerate}

\subsection{Revision of Decision Tree and Working Policy}
\label{sec:tree-revision}

We have seen how $\equiv_{G[U]}$ induces equivalence classes of
vertices. In fact, the function $\pi$ also induces a partitioning of
the vertex set $U$.  Specifically, every $x \in \Ran{\pi}$ defines a
vertex partition
$\mathcal{C}(x) = \{ v \in \Dom{\pi} \mid \pi(v) = x \}$.  It is
intended that the vertex partitions induced by $\pi$ are identical to
the equivalence classes induced by $\equiv_{G[U]}$. Now suppose the
\textsc{NVQ} on line \ref{line:cons-NVQ} returns a novel vertex
$u$. This means $u$ is not equivalent to any previously seen vertex
(second implication of Proposition \ref{prop:indist_vertices}).
Consequently, when the decision tree $\mathcal{T}$ classifies $u$ to a
previously seen vertex $w$ on line \ref{line:cons-classify}, the
classification is incorrect.  In other words, the vertex partitions
induced by $\pi' = \pi[u\mapsto w]$ (line \ref{line:cons-extend-pi})
becomes ``out of sync'' with the equivalence classes induced by
$\equiv_{G[U]}$.  Not only that, the digraph $H$ is no longer a
summary of $G[U]$ after the novel vertex is added to $U$.  Such
discrepancies will be detected on line \ref{line:cons-detect-error}
and then fixed on lines
\ref{line:cons-revise-tree}--\ref{line:cons-revise-H}. After that, the
\textsc{HTQ} on line \ref{line:cons-confirm-revision} will return an
empty set of errors. A detailed exposition of lines
\ref{line:cons-revise-tree}--\ref{line:cons-revise-H} is given below.

\subsubsection{Revision in Algorithm \ref{algo:conservative-learner}}
Line \ref{line:cons-revise-tree} of Algorithm
\ref{algo:conservative-learner} invokes the subroutine \textsc{Revise}
to fix the decision tree $\mathcal{T}$ and the domain assignment
$\pi$. As a result the vertex partitions induced by $\pi$ will be
``synchronized'' with the equivalence classes induced by $G[U]$. (A
detailed explanation of \textsc{Revise} will be given in \S
\ref{sec:revision-in-revise}.)

With $\mathcal{T}$ and $\pi$ now fixed, lines
\ref{line:cons-recompute-V-of-H}--\ref{line:cons-revise-H} revise $H$
so that it is a summary of $G[U]$.  The new vertex set $V$ is simply
the range of the updated domain assignment $\pi$ (line
\ref{line:cons-recompute-V-of-H}). Since $V \subseteq U$, line
\ref{line:cons-recompute-E-of-H} sets the edge set $E$ to contain the
edges in $G[U]$ among the vertices in $V$.  Given the conservatively
extended policy $(H, \pi')$ and its error set $\mathcal{E}$, Algorithm
\ref{algo:edge} is invoked to check if an edge $(u, a, v)$ is in
$G[U]$. Note that no invocation of the \textsc{CNQ} is involved
here. An edge $(u,a,v)$ is in $G[U]$ if and only if either (a) policy
$(H, \pi')$ grants request $(u,a,v)$ and $(u,a,v)$ is not an error, or
(b) policy $(H, \pi')$ denies request $(u,a,v)$ and $(u,a,v)$ is an
error.  The check is expressed as an exclusive-or in Algorithm
\ref{algo:edge}.

\begin{algorithm}[t]
  \caption{\textsc{Revise}($\mathcal{T}$,
    $H$, $\pi$, $u$, $w$, $\mathcal{E}$)\label{alg:revise-tree}}
  \KwIn{$H$, $\pi$, and $\mathcal{T}$ satisfy \textbf{INV-1},
    \textbf{INV-2} and \textbf{INV-3} (see \S
    \ref{sec:conservative-invariant} for the definition of these
    conditions).  Then \textsc{NVQ} returns $u$, $\mathcal{T}$
    classifies $u$ to $w \in V(H)$, and
    $\textsc{HTQ}(H,\,\pi[u\mapsto w])$ returns a non-empty set
    $\mathcal{E}$.  } \KwOut{$(\mathcal{T}^\circ, \pi^\circ)$, where
    $\pi^\circ$ and $\mathcal{T}^\circ$ are updated versions of $\pi$
    and $\mathcal{T}$ that satisfy \textbf{INV-2}(a) and
    \textbf{INV-3}.}  let $\pi' = \pi[u \mapsto w]$ and
  $\pi^\circ = \pi'$\label{line:revise-init-first}\; let
  $\mathcal{T}^\circ = \mathcal{T}$\label{line:revise-init-tree}\; let
  $\mathit{WL}$ be the set of all leaves in
  $\mathcal{T^\circ}$\label{line:revise-init-last}\;
  \While{$\mathit{WL} \neq \emptyset$\label{line:rev-while}}{
    remove a
    leaf $n$ from $\mathit{WL}$\label{line:revise-pop-leaf}\;
    $\mathcal{C} = \{\, v \in \Dom{\pi^\circ} \mid \pi^\circ(v) =
    \ell(n)\,\}$\label{line:revise-org-part}\;
    $t = \textsc{null}$\;
    \If{$\exists v_1, v_2 \in \mathcal{C} \,.\, v_1 \neq v_2 \land
      \linebreak
      \mbox{\qquad} (v_1, a, u) \in \mathcal{E} \land (v_2,
      a, u) \not\in \mathcal{E}$\label{line:revise-check-to}}{
      $t = \mathit{to}(a,u)$\label{line:revise-test-to}\;
      $V^+ = \{ v \in \mathcal{C} \mid \textsc{Edg}(v, a, u, H,
      \pi', \mathcal{E}) \}$\label{line:revise-part-to-plus}\;
      $V^- = \{ v \in \mathcal{C} \mid \lnot \textsc{Edg}(v, a, u, H,
      \pi', \mathcal{E}) \}$\label{line:revise-part-to-minus}\; }
    \ElseIf{$\exists v_1, v_2 \in \mathcal{C} \,.\, v_1 \neq v_2 \land
      \linebreak
      \mbox{\qquad\qquad} (u,a,v_1) \in \mathcal{E} \land
      (u,a,v_2) \not\in \mathcal{E}$\label{line:revise-check-from}}{
      $t = \mathit{from}(u,a)$\label{line:revise-test-from}\;
      $V^+ = \{ v \in \mathcal{C} \mid \textsc{Edg}(u, a, v, H,
      \pi', \mathcal{E}) \}$\label{line:revise-part-from-plus}\;
      $V^- = \{ v \in \mathcal{C} \mid \lnot \textsc{Edg}(u, a, v, H,
      \pi', \mathcal{E}) \}$\label{line:revise-part-from-minus}\;
    }
    \ElseIf{$u \in \mathcal{C} \land (u,a,u) \in \mathcal{E} \land
      \linebreak
      \mbox{\qquad\qquad} \exists v \in \mathcal{C} \,.\,
      (v,a,v) \not\in \mathcal{E}$\label{line:revise-check-loop}}{
      $t = \mathit{loop}(a)$\label{line:revise-test-loop}\;
      $V^+ = \{ v \in \mathcal{C} \mid \textsc{Edg}(v, a, v, H,
      \pi', \mathcal{E}) \}$\label{line:revise-part-loop-plus}\;
      $V^- = \{ v \in \mathcal{C} \mid \lnot \textsc{Edg}(v, a, v, H,
      \pi', \mathcal{E}) \}$\label{line:revise-part-loop-minus}\;
    }
    \If{$t \neq \textsc{null}$}{
      modify $\mathcal{T}^\circ$ by
      replacing leaf $n$ with a decision node $n'$, so that
      $\mathit{test}(n') = t$, $\mathit{left}(n') = n^+$, and
      $\mathit{right}(n') = n^-$\label{line:revise-update-T}\;
      set $\ell(n^+)$ to be some vertex $v^+$ in
      $V^+$\label{line:revise-set-leaf-plus}\;
      set $\ell(n^-)$ to be
      some vertex $v^-$ in $V^-$\label{line:revise-set-leaf-minus}\;
      update $\pi^\circ$ so that $\pi^\circ(x) = v^+$ for $x \in V^+$,
      and $\pi^\circ(x) = v^-$ for $x \in V^-$, and $\pi^\circ(x)$ is
      unchanged if
      $x \not\in \mathcal{C}$\label{line:revise-update-pi}\;
      $\mathit{WL} = \mathit{WL} \cup \{ n^+, n^-
      \}$\label{line:revise-push-leaf}\;
    }
  }
  \Return $(\mathcal{T}^\circ, \pi^\circ)$\;
\end{algorithm}

\subsubsection{Revision in Algorithm \ref{alg:revise-tree}}
\label{sec:revision-in-revise}
Algorithm \ref{alg:revise-tree} is designed to revise
$\pi' = \pi[u\mapsto w]$ to a new function $\pi^\circ$ so that
$\pi^\circ$ and $\equiv_{G[U]}$ are synchronized again. Along the way,
$\mathcal{T}$ is updated to a new decision tree $\mathcal{T}^\circ$
that produces the same classification as $\pi^\circ$.  Let us examine
Algorithm \ref{alg:revise-tree} line by line.

Initially, $\pi^\circ = \pi[u\mapsto w]$ and
$\mathcal{T}^\circ = \mathcal{T}$ (lines
\ref{line:revise-init-first}--\ref{line:revise-init-tree}).
Propositions \ref{prop:dist_vertices} and \ref{prop:indist_vertices}
tell us that, while some vertex partitions induced by $\pi^\circ$
remain identical to equivalence classes induced by $\equiv_{G[U]}$,
other vertex partitions become the union of multiple equivalence
classes. In particular, the equivalence class of the novel vertex $u$
is a singleton set, and it is a proper subset of $\mathcal{C}(w)$.
Algorithm \ref{alg:revise-tree} revises $\pi^\circ$ incrementally. In
each iteration, a vertex partition
$\mathcal{C} = \mathcal{C}(\ell(n))$ for some leaf $n$ is considered
(line \ref{line:revise-org-part}).  The algorithm attempts to detect
if $\mathcal{C}$ contains two distinguishable vertices $v_1$ and
$v_2$.  It does so by detecting a discrepancy in adjacency: e.g., one
of $(v_1, a, u)$ or $(v_2, a, u)$ belongs to $E(G[U])$ but not both.
If such a distinguishable pair exists in $\mathcal{C}$, then
$\mathcal{C}$ is split into two non-empty partitions $V^+$ and $V^-$
(lines
\ref{line:revise-part-to-plus}--\ref{line:revise-part-to-minus},
\ref{line:revise-part-from-plus}--\ref{line:revise-part-from-minus},
and
\ref{line:revise-part-loop-plus}--\ref{line:revise-part-loop-minus}),
and $\pi^\circ$ is updated to reflect this split (line
\ref{line:revise-update-pi}). This brings the partitions induced by
$\pi^\circ$ one step closer to mirroring the equivalence classes
induced by $\equiv_{G[U]}$.

Note that the detection of discrepancies in adjacency does not rely on
issuing \textsc{CNQ}s. Instead, they are discovered by recognizing
\emph{discrepancies in errors} (lines \ref{line:revise-check-to},
\ref{line:revise-check-from}, \ref{line:revise-check-loop}).  For
example, if exactly one of $(v_1, a, u)$ and $(v_2, a, u)$ is in
$E(G[U])$ (a discrepancy in adjacency), then exactly one of
$(v_1, a, u)$ and $(v_2, a, u)$ is in $\mathcal{E}$ (a discrepancy of
errors). This explains why the check on line
\ref{line:revise-check-to} is designed as such.

The decision tree $\mathcal{T}^\circ$ is also updated so that it
produces the same classification as $\pi^\circ$.  Specifically, when
$\mathcal{C}$ is split, the corresponding leaf in $\mathcal{T}^\circ$
is turned into a decision node with two children leaves (lines
\ref{line:revise-update-T}--\ref{line:revise-set-leaf-minus}).  The
test $t$ of the new decision node $n'$ is selected to reflect the way
$\mathcal{C}$ is partitioned into $V^+$ and $V^-$ (lines
\ref{line:revise-test-to}, \ref{line:revise-test-from}, and
\ref{line:revise-test-loop}).

Algorithm \ref{alg:revise-tree} maintains a \Dfn{work list}
$\mathit{WL}$ that tracks vertex partitions that could potentially be
split.  More precisely, $\mathit{WL}$ contains a leaf $n$ if and only
if the partition $\mathcal{C}(\ell(n))$ is a candidate for splitting.
Initially, $\mathit{WL}$ contains all leaves (line
\ref{line:revise-init-last}). One leaf is removed for consideration in
each iteration (line \ref{line:revise-pop-leaf}).  If new leaves are
produced due to splitting, they are added to $\mathit{WL}$ (line
\ref{line:revise-push-leaf}).  The algorithm terminates when the
work list $\mathit{WL}$ becomes empty (line \ref{line:rev-while}).

\subsection{Successful Learning}
\label{sec:conservative-invariant}

We are now ready to demonstrate that the Conservative Learner
(Algorithm \ref{algo:conservative-learner}) is successful (Definition
\ref{def:successful}).  We begin by stating the loop invariants of the
main while-loop (line \ref{line:cons-main-loop}). In the following,
$G$ is the countably infinite digraph encapsulated behind the teacher,
and $U$ is the set of vertices that that have been returned through
\textsc{NVQ}s so far.  (Note that $G[U]$ is finite even though $G$ is
infinite.)

\textbf{INV-1.} $H$ is both a summary and a subgraph of $G[U]$.  (In
other words,
%$H$ is the subgraph constructed in the proof of Lemma
%\ref{lem:isomorphic} in Appendix \ref{app:summary}, in which
$V(H)$ is the set of representatives of the equivalence classes.)
  
\textbf{INV-2.} The domain assignment function $\pi$
satisfies the following conditions: (a) for every $u,v \in U$,
$u \equiv_{G[U]} v$ if and only if $\pi(u) = \pi(v)$; (b)
$\Ran{\pi} = V(H)$. (In English, $\pi$ maps a vertex $v \in U$ to the
representative of the equivalence class to which $v$ belongs.)
  
\textbf{INV-3.} $\mathcal{T}$ is a decision tree for $G$ such that (a)
for every $v \in U$, $\textsc{Classify}(\mathcal{T}, v) = \pi(v)$, and
(b) the number of leaves in $\mathcal{T}$ is $|\Ran{\pi}|$. (In
English, $\mathcal{T}$ and $\pi$ provide the same classification for
vertices in $U$, and each leaf corresponds to a representative.)

We need a technical lemma concerning the correctness of
\textsc{Revise} before we can establish that the conditions above are
indeed the loop invariants of Algorithm
\ref{algo:conservative-learner}.

\begin{lemma}
  \label{lemma:tree-revision}
  Suppose $H$, $\pi$, and $\mathcal{T}$ satisfy \textbf{INV-1},
  \textbf{INV-2}, and \textbf{INV-3} at the beginning of the
  while-loop of Algorithm \ref{algo:conservative-learner}.  Suppose
  further that vertex $u$ is returned from \textsc{NVQ} on line
  \ref{line:cons-NVQ}, $\textsc{Classify}(\mathcal{T}, u)$ returns
  representative $w$ on line \ref{line:cons-classify}, and
  $\textsc{HTQ}(H, \pi[u\mapsto w])$ returns a non-empty set
  $\mathcal{E}$ of errors on line \ref{line:cons-first-confirm}.  Then
  $\textsc{Revise}(\mathcal{T}, H, \pi, \linebreak u, w, \mathcal{E})$
  returns a pair $(\mathcal{T}, \pi)$ that satisfy \textbf{INV-2}(a)
  and \textbf{INV-3}.
\end{lemma}

\begin{proof}
We claim that the following are loop invariants for the while-loop
(line \ref{line:rev-while}) in Algorithm \ref{alg:revise-tree}.
\begin{itemize}
\item[\textbf{REV-1.}] For all $v_1,v_2\in U$, if $v_1 \equiv_{G[U]} v_2$ then
  $\pi^\circ(v_1) = \pi^\circ(v_2)$. (Equivalently,
  $\pi^\circ(v_1) \neq \pi^\circ(v_2)$ implies
  $v_1 \not\equiv_{G[U]} v_2$.)
\item[\textbf{REV-2.}] For every $v \in \Ran{\pi^\circ}$, $\pi^\circ(v) = v$.
\item[\textbf{REV-3.}] If a leaf $n$ of $\mathcal{T}^\circ$ is not in
  $\mathit{WL}$, then $v_1 \equiv_{G[U]} v_2$ for every
  $v_1, v_2 \in \mathcal{C}(\ell(n))$.
\item[\textbf{REV-4.}]
  $\textsc{Classify}(\mathcal{T^\circ}, v) = \pi^\circ(v)$.
\item[\textbf{REV-5.}] The leaf label function $\ell(\cdot)$ is a bijection.
  (That is, $\mathcal{T}^\circ$ has exactly one leaf for
  each vertex in \Ran{ \pi^\circ }.)
\end{itemize}
It is easy to see that, when the loop terminates, if $\mathcal{T}$ and
$\pi$ are updated to $\mathcal{T}^\circ$ and $\pi^\circ$, then
\textbf{REV-1} and \textbf{REV-3} imply \textbf{INV-2}(a), and
\textbf{REV-4} and \textbf{REV-5} entail \textbf{INV-3}.  Note also
that the loop is guaranteed to terminate within $2m$ iterations, where
$m$ is the number of equivalence classes induced by
$\equiv_{G[U]}$. This is the consequence of two observations.  First,
\textbf{REV-1} implies that the number of vertex partitions induced by
$\pi$ is always no larger than the number of equivalence classes
induced by $\equiv_{G[U]}$.  Thus, a vertex partition induced by
$\pi^\circ$ cannot be split indefinitely.  Second, when a vertex
partition $\mathcal{C}$ is selected to be examined in an iteration, if
it is not split during that iteration, then it will be removed
permanently from work list $\mathit{WL}$. Termination follows from
these two observations.  In summary, demonstrating that the above
conditions are loop invariants is sufficient for establishing the
theorem.

We now proceed to show that (a) the invariants are established before
the while-loop starts, and (b) the while-loop preserves the
invariants. Checking (a) is straightforward (see lines
\ref{line:revise-init-first}--\ref{line:revise-init-last}). We verify
(b) below.  The preservation of \textbf{REV-2} and \textbf{REV-5}
follows immediately from lines
\ref{line:revise-set-leaf-plus}--\ref{line:revise-update-pi}. We
demonstrate below the preservation of \textbf{REV-1}, \textbf{REV-3},
and \textbf{REV-4}.

\emph{Preservation of \textbf{REV-1}.}  Suppose the vertex partition
$\mathcal{C}$ in line \ref{line:revise-org-part} is split into
$\mathcal{C}(v^+) = V^+$ and $\mathcal{C}(v^-) = V^-$ on line
\ref{line:revise-update-pi}. (Note that the checks on lines
\ref{line:revise-check-to}, \ref{line:revise-check-from}, and
\ref{line:revise-check-loop} ensure that both $V^+$ and $V^-$ are
non-empty.)  We want to show that $v_1 \not\equiv_{G[U]} v_2$ for
every $v_1 \in V^+$ and $v_2 \in V^-$.  There are three cases. First,
if $V^+$ and $V^-$ were constructed on lines
\ref{line:revise-part-to-plus} and \ref{line:revise-part-to-minus},
then $(v_1, a, u) \in E(G[U])$ but $(v_2, a, u) \not\in
E(G[U])$. Second, $V^+$ and $V^-$ were constructed on lines
\ref{line:revise-part-from-plus} and
\ref{line:revise-part-from-minus}, resulting in
$(u, a, v_1) \in E(G[U])$ but $(u, a, v_2) \not\in E(G[U])$. Third,
$V^+$ and $V^-$ were constructed on lines
\ref{line:revise-part-loop-plus} and
\ref{line:revise-part-loop-minus}, and thus
$(v_1, a, v_1) \in E(G[U])$ but $(v_2, a, v_2) \not\in E(G[U])$.  In
each case, $v_1 \not\equiv_{G[U]} v_2$.  %%% Definition def:indistinct

\emph{Preservation of \textbf{REV-3}.} Suppose $n$ is a leaf in
$\mathcal{T}$ but $n \in \mathit{WL}$ at the beginning of an
iteration.  Suppose further that $n$ remains a leaf of $\mathcal{T}$
but $n \not\in \mathit{WL}$ at the end of that iteration. This happens
because the vertex partition $\mathcal{C}$ (line
\ref{line:revise-org-part}) is not split during the iteration, meaning
all the three checks on lines \ref{line:revise-check-to},
\ref{line:revise-check-from}, and \ref{line:revise-check-loop} were
negative.  By way of contradiction, assume there exists
$v_1, v_2 \in \mathcal{C}$ such that $v_1 \not\equiv_{G[U]}
v_2$. There are now two cases.

\textbf{Case 1}: \emph{neither $v_1$ nor $v_2$ is $u$}.  Adjacency
among vertices existing before the introduction of $u$ remains
unchanged. Thus, condition (2) in the definition of
$\equiv_{G[U]}$  %%% Definition def:indistinct
must
have been violated by a discrepancy between
$\mathit{adj}_{G[U]}(v_1, u)$ and $\mathit{adj}_{G[U]}(v_2, u)$.  This
discrepancy leads to errors in $\mathcal{E}$ that are picked
up on either line \ref{line:revise-check-to} or line
\ref{line:revise-check-from}, contradicting the fact that
no splitting occurs in this iteration.

\textbf{Case 2}: \emph{one of $v_1$ or $v_2$ is $u$}. (Without loss of
generality, assume $v_2 = u$.)  Again, adjacency among old vertices
remain unchanged.  Thus, condition (1) in the definition
of $\equiv_{G[U]}$  %%% Definition def:indistinct
must have been violated between $v_1$ and
$u$. This produces errors that should have been picked up by one of
the tests on lines \ref{line:revise-check-to},
\ref{line:revise-check-from}, and \ref{line:revise-check-loop},
contradicting the fact no splitting occurs in this iteration.

\emph{Preservation of \textbf{REV-4}.} Suppose $\pi^\circ$ and
$\mathcal{T}^\circ$ produce the same classification in the beginning
of an iteration. Suppose $\pi^\circ$ and $\mathcal{T}^\circ$ are
updated in lines
\ref{line:revise-update-T}--\ref{line:revise-update-pi}. The updated
$\pi^\circ$ and $\mathcal{T}^\circ$ still return the same
classification only if the choice of test $t$ is consistent with the
partitioning of $\mathcal{C}$ into $V^+$ and $V^-$. A careful
examination of lines
\ref{line:revise-test-to}--\ref{line:revise-part-to-minus},
\ref{line:revise-test-from}--\ref{line:revise-part-from-minus}, and
\ref{line:revise-test-loop}--\ref{line:revise-part-loop-minus}
confirms this.
\end{proof}

\begin{theorem}
  \label{thm:conservative-loop-invariants}
  Conditions \textbf{INV-1}, \textbf{INV-2}, and \textbf{INV-3} are
  loop invariants of the main while-loop
  (line \ref{line:cons-main-loop})
  in Algorithm \ref{algo:conservative-learner}.
\end{theorem}
\begin{proof}
  We demonstrate two points regarding the loop invariants (\S
  \ref{sec:conservative-invariant}) of the Conservative Learner
  (Algorithm \ref{algo:conservative-learner}): (1) the loop invariants
  are established prior to the entrance of the while-loop; (2) the
  while-loop preserves the loop invariants.

\emph{(1) Initialization.} After the first vertex $u$ is returned by the
\textsc{NVQ} on line \ref{line:cons-initial-NVQ}, we have
$U = \{ u \}$. Lines
\ref{line:cons-init-vertices}--\ref{line:cons-init-edges} initialize
$H$ to be $G[U]$ by asking \textsc{CNQ}s.  \textbf{INV-1} is therefore
established.  Then \textbf{INV-2} is established on line
\ref{line:cons-init-pi} by initializing $\pi$ such that
$\Dom{\pi} = \{u\}$ and $\pi(u) = u$. Lastly, line
\ref{line:cons-init-tree} establishes \textbf{INV-3}.  All invariants
are thus established by the time the while-loop is entered.

\emph{(2) Preservation.}  We demonstrate that, if the three invariants hold
at the beginning of an iteration, then they still hold by the end of
that iteration.

Suppose the three invariants hold at the beginning of an iteration.
Line \ref{line:cons-NVQ} requests a new vertex from the teacher.  The
effect is that $G[U]$ now has an extra vertex.  The loop invariants
are invalidated as a consequence.  Algorithm
\ref{algo:conservative-learner} re-establishes the loop invariants
using lines \ref{line:cons-classify}--\ref{line:cons-revise-H}.

In accordance to the Occam's Razor principle, the learner presumes
that $H$ is still a summary of $G[U]$. That assumption holds if $u$ is
indistinguishable from an existing vertex $v$ (Proposition
\ref{prop:indist_vertices}).  Consequently, line
\ref{line:cons-classify} uses the decision tree $\mathcal{T}$ to
obtain a classification for $u$.  Since $u$ is supposed to share the
same adjacency pattern as $v$, $\mathcal{T}$ classifies $u$ to the
representative $w$ of $v$'s equivalence class.  The protection domain
assignment $\pi$ is now updated to $\pi[u \mapsto w]$ (lines
\ref{line:cons-extend-pi} and \ref{line:cons-no-revision}). All these are
done under the assumption of indistinguishability, which is tested on
line \ref{line:cons-first-confirm} by the \textsc{HTQ}.  If the test
results in no errors, then \textbf{INV-1}, \textbf{INV-2}, and
\textbf{INV-3} are re-established.

If the presumption of indistinguishability turns out to be invalid
($\mathcal{E} \neq \emptyset$), then lines
\ref{line:cons-revise-tree}--\ref{line:cons-revise-H} will
re-establish the invariants by recomputing $H$, $\pi$, and
$\mathcal{T}$. This is achieved in two steps.  The first step
corresponds to line \ref{line:cons-revise-tree}, which revises $\pi$
and $\mathcal{T}$ so that \textbf{INV-2}(a) and \textbf{INV-3} are
recovered (Lemma \ref{lemma:tree-revision}).  The second step in
re-establishing the invariants is specified in lines
\ref{line:cons-recompute-V-of-H}--\ref{line:cons-revise-H}, in which
$H$ is recomputed to recover \textbf{INV-1} and \textbf{INV-2}(b).
Specifically, line \ref{line:cons-recompute-V-of-H} takes the range of
function $\pi$ (which are the representatives of equivalence classes)
to be the vertices of $H$. This re-establishes \textbf{INV-2}(b).
Line \ref{line:cons-recompute-E-of-H} then uses the edges in $G[U]$
among the representatives to be the edges of $H$.
\textbf{INV-1} is therefore re-established.
\end{proof}

The loop invariants allow us to deduce that learning
in Algorithm \ref{algo:conservative-learner} proceeds in the
manner prescribed by \textbf{SC-1} and \textbf{SC-2}.
\begin{theorem} \label{thm:conservative-successful}
  The Conservative Learner is successful.
\end{theorem}

\begin{proof}
  We demonstrate \textbf{SC-1} and \textbf{SC-2} in turn.
  
  \textbf{SC-1}: An immediate corollary of \textbf{INV-1} and
  \textbf{INV-2} is that $H$ is irreducible and $\pi$ is
  surjective. In addition, when the \textsc{HTQ} is invoked on lines
  \ref{line:cons-initial-confirm} and
  \ref{line:cons-confirm-revision}, \textbf{INV-1} and \textbf{INV-2}
  hold. What remains to be shown is that $H$ is irreducible and $\pi'$
  is surjective on line \ref{line:cons-first-confirm}. To see this,
  note two facts: (a) prior to the \textsc{NVQ} on line
  \ref{line:cons-NVQ}, $H$ is the summary of $G[U]$ and thus
  irreducible; (b) $\pi$ and
  $\pi'$ has the same range and thus $\pi'$ is also surjective.

  \textbf{SC-2}: Since \textbf{INV-1} and \textbf{INV-2} are already
  established by the time line \ref{line:cons-initial-confirm} is
  reached, the \textsc{HTQ} on line \ref{line:cons-initial-confirm}
  returns an empty set.  What remains to be shown is that, if the
  \textsc{HTQ} on line \ref{line:cons-first-confirm} returns a
  non-empty set of errors, then the \textsc{HTQ} on line
  \ref{line:cons-confirm-revision} returns an empty set.  This, again,
  holds as \textbf{INV-1} and \textbf{INV-2} are re-established by the
  time line \ref{line:cons-confirm-revision} is reached.
  Consequently, \textbf{SC-2} is satisfied.
\end{proof}

\subsection{Administration Cost and Error Bound}
\label{sec:conservative-complexity}

We assess the administration cost incurred by the Conservative Learner
\begin{theorem}
  \label{thm:conservative-admin-cost}
  Let $k = |\Sigma|$, the number of access rights.  Suppose the
  Conservative Learner
  % (Algorithm \ref{algo:conservative-learner})
  has received a set $U$ of $n$ vertices through \textsc{NVQ}s, and
  the equivalence relation $\equiv_{G[U]}$ induces $m$ equivalence
  classes.  Then the learner has invoked the \textsc{CNQ} for no more
  than $k + (n-1)(m-1)$ times.
\end{theorem}
\begin{proof}
  The \textsc{CNQ} is invoked $k$ times on line
  \ref{line:cons-init-edges}. The remaining \textsc{CNQ}s are caused
  by the $(n-1)$ invocations of \textsc{Classify} on line
  \ref{line:cons-classify}.  Since the decision tree $\mathcal{T}$ has
  at most $m$ leaves (\textbf{INV-3}), the number of decision nodes in
  $\mathcal{T}$ is no more than $(m-1)$. Thus no more than $(m-1)$
  \textsc{CNQ}s are issued each time \textsc{Classify} is invoked.
  The total number of \textsc{CNQ}s is no more than $k + (n-1)(m-1)$.
\end{proof}

While $k$ is a constant, the term $(n-1)(m-1)$ grows linearly to both
$m$ and $n$. If $n \gg m$, meaning the number of entities grows much
faster than the number of protection domains, then the bound above
represents a significant improvement over the quadratic bound ($kn^2$)
of the Tireless Learner.  If $m$ is bounded by a constant (i.e., the
number of protection domains is fixed), then the improvement is even
more prominent.

This reduction in administration cost is nevertheless achieved by
tolerating errors.
\begin{theorem}
  \label{thm:conservative-errors}
  Let $k = |\Sigma|$, the number of access rights.  Suppose the
  Conservative Learner
  % (Algorithm \ref{algo:conservative-learner})
  has received a set $U$ of $n$ vertices through \textsc{NVQ}s, and
  the equivalence relation $\equiv_{G[U]}$ induces $m$ equivalence
  classes.  Then the learner has committed no more than $k(2n+1)(m-1)$
  errors.
\end{theorem}
\begin{proof}
  In the proof of Theorem \ref{thm:conservative-successful}, we
  observe that only the \textsc{HTQ} on line
  \ref{line:cons-first-confirm} can return a non-empty set
  $\mathcal{E}$ of errors.  This occurs when the \textsc{NVQ} on line
  \ref{line:cons-NVQ} returns a novel vertex (see the second
  implication of Proposition \ref{prop:indist_vertices}).  Novel
  vertices are returned no more than $(m-1)$ times as there are at
  most $m$ equivalence classes.  Suppose the $i$th vertex returned by
  a \textsc{NVQ} is a novel vertex $u$. The size of $|\mathcal{E}|$ is
  at most $k(2i-1)$. The reason is that there is at most $k$ errors of
  the form $(u,a,u)$, $k(i-1)$ errors of the form $(x,a,u)$, and
  $k(i-1)$ errors of the form $(u,a,x)$.  Thus the later a novel
  vertex is returned by an \textsc{NVQ}, the bigger $|\mathcal{E}|$
  will be.  The worst case is when the last $(m-1)$ invocations of the
  \textsc{NVQ} all return novel vertices. The overall number of errors
  will be at most
  \[
    \sum^{n}_{i=n-(m-1)+1} k(2i-1)
    =
    k(2n-m+1)(m-1)
  \]
  which is smaller than $k(2n+1)(m-1)$ as required.
\end{proof}

Note that the number of errors is also linear to both $n$ and $m$. The
typical case, again, is either $m$ grows much slower than $n$ or $m$
is bounded by a constant.

Compared to the Tireless Learner, which avoids errors at all cost, the
Conservative Learner offers a much lower administration cost (linear
rather than quadratic), but does so by allowing linearly many
errors. We have therefore demonstrated that \emph{the cost of policy
  administration can be reduced if appropriate heuristic reasoning is
  employed.}

The benefits of adopting Occam's Razor (assuming a vertex is not novel
until errors prove otherwise) can be put into sharper focus when we
impose a probabilistic distribution over how the teacher chooses
vertices to be returned.  Suppose there are at most $m$ equivalence
classes and that each time the teacher returns a vertex, the selection
is independent of previously returned vertices.  Suppose further that
$p_i$ is the probability that the teacher chooses a vertex from the
$i$'th equivalence class to be returned to the learner.  Here,
$\sum_{i=1}^{m} p_{i} = 1$.  We are interested in knowing the expected
number of \textsc{NVQ} invocations required for the learner to have
sampled at least one vertex from every equivalence class.  This number
is significant for the Conservative Learner, because after having seen
a representative vertex from each equivalence class, the rest of the
learning process will be error free, involving only the classification
of vertices into existing equivalence classes.

The problem above is in fact an instance of the \textbf{coupon
  collector problem} \cite[Ch.~7]{Ross2012}. Let random variable
$X_{i}$ be the number of \textsc{NVQ}'s the learner has issued before
a first vertex from the $i$'th equivalence class is returned. Then
$X = \max(X_{1}, \ldots ,X_{m})$ is the number of \textsc{NVQ}'s issued
before at least one vertex from each equivalence class is returned.
According to the formula of coupon collecting with unequal
probabilities \cite[Ch.~7]{Ross2012},
\[
E[X] = \sum_{i=1} \frac{1}{p_{i}} - \sum_{i<j} \frac{1}{p_{i}+p_{j}} + \cdots + (-1)^{m+1} \frac{1}{p_{1}+ \cdots +p_{m}}
\]

\begin{comment}
Since that
\[
\int_{0}^{\infty} e^{-px} dx = \frac{1}{p}
\]
We get the formula in the form of an integral. 
\[
E[X] = \int_{0}^{\infty} (1-\prod_{i=1}^{n}(1-e^{-p_{i}x})) dx
\]
\end{comment}

Consider the special case when the vertices of each equivalence class have an
equal probability of being chosen by the teacher.  In other words,
$p_{1} = p_{2} = \cdots = p_{m} = \frac{1}{m}$.
\[
E[X] = m \sum_{i=1}^{m} \frac{1}{i} = m \cdot H_{m}
\]
where $H_{m}$ is the $m$'th harmonic number.
Employing the well-known approximation for the harmonic series
\cite[App.~1]{Cormen2009}, we get
\[
E[X] \approx m \ln m
\]
In the average case, the Conservative Learner cumulates errors only in
the first $m\ln m$ rounds of learning. After that, learning involves
only the error-free classification of vertices into existing
equivalence classes. Remarkably, $E[X]$ depends only on $m$.

In conclusion, a key insight offered by this section is the following:
\emph{Fallible, heuristic reasoning is the source of scalability in
  policy administration. An access control model scales better than
  the access control matrix because it provides conceptualizing
  instruments (e.g., protection domains, roles, attributes,
  relationships) that support heuristic reasoning without producing
  too many errors.}

\section{Related Work} \label{sec:related-work}
\subsection{Active Learning}
In active learning \cite{Settles2012}, a
learner actively formulates queries and directs them to a teacher (an
oracle), whose answers allow the learner to eventually learn a
concept.  In computational learning theory \cite{Kearns1994}, active
learning is studied in a formal algorithmic framework, in which the
learning algorithm is evaluated by its query complexity (i.e., the
number of queries required for successful learning).  We use active
learning as a framework for constructing computational models of
policy administration, so that the cost of policy administration can
be quantified in terms of query complexity. 

Angluin proposes the exact learning algorithm $L^*$ for learning
finite automata \cite{Angluin1987}. Her learning protocol involves two
queries: (i) the membership query, in which the learner asks if a
certain string is in the target language, and (ii) the equivalence
query, by which the learner asks if a concrete finite automaton is
equivalent to the target concept.  The equivalence query returns a
counterexample if the answer is negative.  A well-known variation of
$L^*$ is the algorithm of Kearns and Vazirani \cite{Kearns1994}, which
employs a decision tree as an internal data structure for classifying
strings. The design of our learning protocol is influenced by the
Angluin learning model: \textsc{CNQ} and \textsc{HTQ} play a role
analogous to that of membership and equivalence query.  Our use of
decision trees has been inspired by the algorithm of Kearns and
Vazirani \cite{Kearns1994}.  Our learning model is nevertheless
distinct from previous work, in at least three ways: (a) our goal is
to learn a digraph summary and its corresponding strong homomorphism,
(b) as the encapsulated digraph is infinite, the learner is modelled
as a reactive process, the convergence of which is formalized in
\textbf{SC-2}, and (c) we formulate queries to model entity
introduction (\textsc{NVQ}), policy deliberation (\textsc{CNQ}), and
policy assessment (\textsc{HTQ}).

Also related is Angluin's later work on learning hidden graphs
\cite{Angluin2008, Reyzin2007}. The edges of a finite graph are hidden
from the learner, but its vertices are fully known. The learner
employs a single type of queries (such as edge detection queries or
edge counting queries) to recover the edges via a Las Vegas algorithm.
Our work, again, is different. Not only is our hidden graph infinite,
we are learning a digraph summary rather than all the edges.  Also,
ours is an exact learning model (\textsc{SC-2}), while theirs is a
probabilistic one.

\subsection{Policy Mining}
As access control models are being adopted in increasingly complex
organizational settings, the formulation of access control policies
sorely needs automation.  Policy mining is about the inference of
policies from access logs.  The increase of scale in IoT systems only
makes the need for policy mining more acute.
Role mining \cite{Vaidya2007, Mitra2016} is concerned with the
automated discovery of RBAC roles using matrix decomposition.  The
problem itself is $\mathit{NP}$-hard. A sample of research in this
direction includes \cite{Frank2009, Frank2013, Molloy2010, Xu2012,
  Xu2013}.  ABAC policy mining is also $\mathit{NP}$-hard
\cite{Xu2015}.  Representative works include \cite{Xu2015, Medvet2015,
  Karimi2018, Cotrini2018, Iyer2018}.  The mining of ReBAC policies is
studied in \cite{Bui2017, Bui2019C, Bui2019J, Iyer2019}. 

Particularly related to our work is that of Iyer and Masoumzadeh
\cite{Iyer2020}, who adopted Angluin's $L^{*}$ algorithm for learning
ReBAC policies, which are represented as deterministic finite
automata.  Their algorithm used a mapper component to accept
relationship patterns from learners and reply to access decisions by
interacting with the policy decision point (PDP).  The mapper is an
additional component between the learner and the teacher.  The
learning algorithm (as the learner) takes only relationship patterns
as input.  The PDP (as the teacher) takes only access requests as
input.  The mapper translates relationship patterns into access
requests.  Then interacts with the PDP which determines access
decisions for given access requests.  Relationship patterns are
sequences of relationship labels that are expressed in ReBAC policies.

Recently, utilizing ML models to assist in mitigating administration
costs resulting from policy changes was studied in \cite{Nobi2022}.
It demonstrated that ML models such as a random forest or a residual
neural network (ResNet) are both feasible and effective in adapting to
new changes in MLBAC administration.  

This work is \emph{not} about policy mining.  Instead, this work uses
active learning as a framework to model the \emph{human} process of
policy administration.  A learner, even though specified
algorithmically, is a computational model of the policy administrator
(a human).  This modelling approach allows us to quantify the
cognitive efforts carried out by the policy administrator as she
evolves the access control policy over time. Armed with this
quantification method, we can now compare different policy
administration strategies.

\section{Conclusion and Future Work} \label{sec:conclusion}
We developed a computational model for the policy administration
process. Specifically, ongoing policy deliberation and revision are
modelled as active learning. The goal is to quantify the cost of
policy administration. We applied this modelling framework to study
the administration of domain-based policies.  We deployed the
aforementioned active learning framework to study how a policy
administrator evolves a domain-based policy to account for the
incremental introduction of new entities.  Two important insights
emerge from this work:
\begin{enumerate}
\item The cost of policy administration depends not only on the choice
  of access control model, but also on the adoption of a complementary
  policy administration strategy.
\item The source of scalability of a policy administration strategy
  comes from its adoption of appropriate learning heuristics.  The
  latter, though fallible, lower administration cost by allowing a
  small number of errors and providing mechanisms to fix the
  policy when errors are detected.
\end{enumerate}
This work therefore suggests a novel methodology for future research
to substantiate, in a quantitative manner, a claim that a given access
control model reduces the cost of policy administration:
\begin{enumerate}
\item Devise an active learning framework for the access control model
  in question (e.g., ABAC, ReBAC, etc).  The querying protocol shall
  capture several aspects of reality: (a) the introduction of new
  entities (or other forms of organizational changes), (b) queries
  that correspond to policy deliberation, and (c) the assessment of a
  candidate policy in terms of errors.
\item Develop a learner that embodies a certain heuristic policy
  administration strategy.
\item Demonstrate that the policy maintained by the learner
  ``converges'' to the actual policy it is trying to learn.
\item Assess the policy administration cost as well as the errors.
  Demonstrate that the administration cost is lower than a certain
  baseline ($kn^2$ in the case of access control matrix).
\end{enumerate}

Several future research directions present themselves.
\begin{inparaenum}
\item Active learning frameworks for other access control paradigms
  (e.g., ReBAC, ABAC) may allow us to characterize heuristics in
  policy administration strategies and quantify their
  administration costs.
\item How do we formalize cases when the learner has \emph{a priori}
  knowledge of the target policy?
\item Further develop the active learning framework for domain-based
  policies.  As an example, learning criteria less aggressive than
  \textbf{SC-1} and \textbf{SC-2} may allow the learner to lower its
  query complexity by converging more slowly.  Another example:
  Alternative definitions of the \textsc{HTQ} may allow us to study
  other ways to assess policies (e.g., deny
  errors are more tolerated over grant errors).
%\item An alternative teacher model in which the teacher is cooperative 
%  rather than ambivalent may better reflect the cost of policy
%  deliberation.
\end{inparaenum}

\begin{comment}
\begin{acks}
\end{acks}
\end{comment}

\bibliographystyle{abbrv}
\bibliography{references}

\begin{comment}
\appendix
\section{Research Methods}
\end{comment}

\end{document}